\documentclass[a4paper]{article}
\usepackage{mathrsfs}
\usepackage{fullpage,amsmath,amsfonts,amsthm,pmat}
\usepackage{amssymb}
\usepackage{cite}
\usepackage{showlabels}
\newcommand\sA{\mathcal A}
\renewcommand\tilde\widetilde
\renewcommand\hat\widehat
\newcommand\Ber{\mathop{\mathrm {Ber}}}
\newcommand\Berstar{\mathop{\mathrm {Ber}^*}}
\newcommand\B{\mathop{\mathrm {B\,\!}}}
\theoremstyle{example}

\theoremstyle{definition}

\newtheorem{rem}{Remark}

\theoremstyle{plain}
\newtheorem{lemma}{Lemma}
\newtheorem{thm}{Theorem}

\begin{document}
    \title{Darboux transformations for a twisted derivation and quasideterminant solutions to the super KdV equation}
    \author{C.X. Li$^{1}$\footnote{trisha\_li2001@yahoo.com} and J.J.C. Nimmo$^{2}$\footnote{j.nimmo@maths.gla.ac.uk}\\$^{1}$School of Mathematical Sciences,\\
    Capital Normal University,
    Beijing 100048, CHINA\\
     $^{2}$Department of Mathematics,\\
    University of Glasgow,
    Glasgow G12 8QW, UK
    }
    \date{}
\maketitle

\begin{abstract}
This paper is concerned with a generalized type of Darboux transformations defined in terms of a twisted derivation $D$ satisfying $D(AB)=D(A)+\sigma(A)B$ where $\sigma$ is a homomorphism. Such twisted derivations include regular derivations, difference and $q$-difference operators and superderivatives as special cases. Remarkably, the formulae for the iteration of Darboux transformations are identical with those in the standard case of a regular derivation and are expressed in terms of quasideterminants. As an example, we revisit the Darboux transformations for the Manin-Radul super KdV equation, studied in Q.P. Liu and M. Ma\~nas, Physics Letters B \textbf{396} 133--140, (1997). The new approach we take enables us to derive a unified expression for solution formulae in terms of quasideterminants, covering all cases at once, rather than using several subcases. Then, by using a known relationship between quasideterminants and superdeterminants, we obtain expressions for these solutions as ratios of superdeterminants. This coincides with the results of Liu and Ma\~nas in all the cases they considered but also deals with the one subcase in which they did not obtain such an expression. Finally, we obtain another type of quasideterminant solutions to the Main-Radul super KdV equation constructed from its binary Darboux transformations. These can also be expressed as ratios of superdeterminants and are a substantial generalization of the solutions constructed using binary Darboux transformations in earlier work on this topic.
\end{abstract}

\section{Introduction}
There have been many papers written recently on noncommutative versions of soliton equations, such as the KP equation,the KdV equation, the Hirota-Miwa equation and the two-dimensional Toda lattice equation \cite{K,P,S,WW1,WW2,WW3,H,HT1,DH,JN,LN1,CGJN,GN2,LN2}. It has been shown that noncommutative integrable systems often have solutions expressed in terms of quasideterminants \cite{GR}. In \cite{CGJN} for example, two families of solutions of the noncommutative KP equation were presented which were termed quasiwronskians and quasigrammians. The origin of these solutions was explained by Darboux and binary Darboux transformations. The quasideterminant solutions were then verified directly using formulae for derivatives of quasideterminants (see also \cite{DH}).

Supersymmetric integrable systems are a particular noncommutative extension of integrable systems, and have attracted much attention because of their applications in physics. Some well-known example are supersymmetric versions of the KdV, KP, sine-Gordon, nonlinear Schr\"odinger equation, AKNS and Harry Dym equations \cite{ManinRadul,MP,KBA,CMKP,HS,RMKM,MCPL,LQP}. Among these, the Manin-Radul super KdV equation is perhaps the best known and has been studied extensively and a number of interesting properties have been established. We mention here the existence of an infinite number of conservation laws, a bi-Hamiltonian structure \cite{OWPZ}, a bilinear form \cite{CARA,MNYM} and Darboux transformations \cite{HS,QPM,LPMM}.

Partly motivated by the properties of superderivatives, which we describe in Section 2, we consider a generalized derivation which has regular derivations, difference operators, $q$-difference operators and superderivatives as some of its special cases. We call this a twisted derivation, following the terminology used in \cite{HLS,CDC}. We show that one can formulate Darboux transformations for such twisted derivations and the iteration formulae are expressed in terms of quasideterminants in which one simply replaces the derivative with the twisted derivation. A special case of this result can be used to construct solutions to supersymmetric equations in terms of quasideterminants.

In \cite{QPM}, solutions for the Manin-Radul super KdV equation were constructed by iterating Darboux transformations by considering the cases of an even number and an odd number of iterations separately. All but one of the formulae obtained by the authors expressed the solutions in terms of superdeterminants.

In this paper, we use an alternative approach to the Darboux transformations using quasideterminants. This is successful in obtaining unified formulae for the solutions, not depending on the parity of the number of iterations. From these quasideterminant solutions, we are not only able to recover superdeterminant solutions given in \cite{QPM} but also get the superdeterminant representation in the one case they did not.

The paper is organized as follows. In Section 2, we give a brief review on superdeterminants, quasideterminants and the relationship between them. In Section 3, we define a twisted derivation and related Darboux transformation. We obtain a formula for iteration of this twisted Darboux transformation in terms of quasideterminants. Section 4 gives applications of both Darboux and binary Darboux transformations to the Manin-Radul super KdV system. In this section we obtain two general solution formulae in terms of quasideterminants, obtained using iterated Darboux and binary Darboux transformations respectively, and then show how these can be expressed in terms of superdeterminants. On the other hand, we present quasideterminant solutions constructed from its binary Darboux transformations to the MRSKdV system. In Section 5 we present conclusions.

\section{Superdeterminants and quasideterminants}
In this section, we collect together some basic facts about supersymmetric objects such as superderivatives, supermatrices, supertranspose and superdeterminants \cite{BZFA,DWB}, about quasideterminants \cite{GR,GGRW,EGR} and about the relationship between superdeterminants and quasideterminants \cite{BMRJ}. The reader is referred to the above mentioned literature for more details.

\subsection{Superderivatives, supertranspose and superdeterminants}\label{sec:super}
Let $\sA$ be a supercommutative, associative, unital superalgebra over a (commutative) ring $K$. There is a standard $\mathbb Z_2$-grading $\sA=\sA_0\oplus\sA_1$ such that $\sA_i\sA_j\subseteq\sA_{i+j}$. Elements of $\sA$ that belong to either
$\sA_0$ or $\sA_1$ are called \emph{homogeneous}; those in $\sA_0$ are called \emph{even} and those in $\sA_1$ are called \emph{odd}. The \emph{parity} $|a|$ of a homogeneous element $a$ is 0 if it is even and 1 if it is odd. It follows that if $a,b$ are homogeneous then $|ab|=|a|+|b|$. Supercommutativity means that all homogeneous elements $a,b$ satisfy $ba=(-1)^{|a||b|}ab$, i.e. even elements commute with all elements, and odd elements anticommute. In particular, this implies that $a_1^2=0$, for all $a_1\in\sA_1$.

\paragraph{Grade involution and superderivative}
The homomorphism $\hat\ \colon\sA\to\sA$ satisfying $\hat a_i=(-1)^ia_i$ for $a_i\in\sA_i$ is called the \emph{grade involution}. For general $a\in\sA$, expressed as $a=a_0+a_1$ where $a_i\in\sA_i$, we have $\hat a=a_0-a_1$. Also for any matrix $M=(m_{ij})$ over $\sA$, $\hat M:=(\hat m_{ij})$. It is easy to see that $\hat{\hat a}=a$.

A \emph{superderivative} $D$ is a linear mapping $D\colon\sA\to\sA$ such that $D(K)=0$ and $D(\sA_i)\subseteq\sA_{i+1}$
and satisfying $D(ab)=D(a)b+\hat a D(b)$. One way to obtain a superderivative is as $D=\partial_\theta+\theta\partial_x$ where $x$ is an even variable and $\theta$ is an odd (Grassmann) variable. For such a superderivative $D^2=\partial_x$.

Note that since $D(\sA_0)\subseteq\sA_1$ and $D(\sA_1)\subseteq\sA_0$, it follows that $D(\hat a)=D(a_0)-D(a_1)=-\hat{D(a)}$ and so grade involution and superderivatives anticommute.

\paragraph{Even and odd supermatrices}
A block matrix $\mathcal {M}=\begin{pmatrix}X&Y\\Z&T\end{pmatrix}$ over $\sA$ where $X$ is $r\times m$, $Y$ is $r\times n$, $Z$ is $s\times m$ and $T$ is $s\times n$ for integers $r$, $s$, $m$ and $n$ with $r,m\ge1$ and $s,n\ge0$ is called an $(r|s)\times(m|n)$ \emph{supermatrix}. It is said to be a \emph{even}, and has parity $0$, if $X$ and $T$ (if not empty) have even entries and $Y$ and $Z$ (if non-empty) have odd entries. One the other hand, if $X$ and $T$ have odd entries and $Y$, $Z$ have even entries then $\mathcal {M}$ is said to be a \emph{odd}, and has parity $1$. It is said to be homogeneous if it is either even or odd.

\paragraph{Supertranspose}
The \emph{supertranspose} of a homogeneous supermatrix $\mathcal{M}$, is defined to be
\begin{equation}
	\mathcal{M}^{st}=\begin{pmatrix}X^t & (-1)^{|\mathcal{M}|}Z^t\\-(-1)^{|\mathcal{M}|}Y^t&T^t\end{pmatrix},	
\end{equation}
where $^t$ denotes the normal matrix transpose. In particular, an even $(m|n)$-row vector has the form
$
	(a_{01},a_{02},\dots,a_{0m},a_{11},a_{12},\dots,a_{1n}),
$
where $a_{ij}\in\sA_i$, and its supertranspose is
\begin{equation}
(a_{01},a_{02},\dots,a_{0m},a_{11},a_{12},\dots,a_{1n})^{st}=(a_{01},a_{02},\dots,a_{0m},-a_{11},-a_{12},\dots,-a_{1n})^t.
\end{equation}
On the other hand, an odd $(m|n)$-row vector has the form
$
(a_{11},a_{12},\dots,a_{1m},a_{01},a_{02},\dots,a_{0n}),
$
and the supertranspose
\begin{equation}
(a_{11},a_{12},\dots,a_{1m},a_{01},a_{02},\dots,a_{0n})^{st}=(a_{11},a_{12},\dots,a_{1m},a_{01},a_{02},\dots,a_{0n})^t.
\end{equation}

For homogenous supermatrices $\mathcal{L}$, $\mathcal{M}$ and $\mathcal{N}$, it is known that
\begin{align}
\label{st}
(\mathcal{M}\mathcal{N})^{st}&=(-1)^{|\mathcal{M}||\mathcal{N}|}{\mathcal{N}}^{st}{\mathcal{M}}^{st},\\
\label{st st}
(\mathcal{M}^{st})^{st}&=(-1)^{|\mathcal{M}|}\hat{\mathcal{M}},
\end{align}
and
\begin{align}
(\mathcal{L}\mathcal{M}\mathcal{N})^{st}=(-1)^{|\mathcal{L}||\mathcal{M}|+|\mathcal{L}||\mathcal{N}|+|\mathcal{M}||\mathcal{N}|}
\mathcal{N}^{st}\mathcal{M}^{st}\mathcal{L}^{st}.
\end{align}
Also, supertranspose commutes with the grade involution but not with a superderivative; for a homogeneous matrix $\mathcal M$,
\begin{equation}
	\bigl(D(\mathcal M)\bigr)^{st}=(-1)^{|\mathcal M|}D(\hat{\mathcal M}^{st}).
\end{equation}

\paragraph{Superdeterminants} Consider an even $(m|n)\times(m|n)$ supermatrix $\mathcal M=\begin{pmatrix}X&Y\\Z&T\end{pmatrix}$ in which $X$ and $T$ are non-singular. The \emph{superdeterminant}, or \emph{Berezinian}, of $\mathcal{M}$ is defined to be
\[
\Ber(\mathcal{M})=\frac{\det(X-YT^{-1}Z)}{\det(T)}=\frac{\det(X)}{\det(T-ZX^{-1}Y)}.
\]
It is also convenient to define
\[
\Berstar(\mathcal{M})=\frac1{\Ber(\mathcal{M})}.
\]

\subsection{Quasideterminants}
An $n\times n$ matrix $M=(m_{i,j})$ over a ring $\mathcal R$ (noncommutative,
in general) has $n^2$ \emph{quasideterminants} written as
$|M|_{i,j}$ for $i,j=1,\dots, n$, which are also elements of
$\mathcal R$. They are defined recursively by

\begin{align}\label{defn}
    |M|_{i,j}&=m_{i,j}-r_i^j(M^{i,j})^{-1}c_j^i,\quad M^{-1}=(|M|_{j,i}^{-1})_{i,j=1,\dots,n}.
\end{align}
In the above $r_i^j$ represents the $i$th row of $M$ with the $j$th
element removed, $c_j^i$ the $j$th column with the $i$th element
removed and $M^{i,j}$ the submatrix obtained by removing the $i$th
row and the $j$th column from $M$. Quasideterminants can be also
denoted as shown below by boxing the entry about which the expansion
is made
\[
|M|_{i,j}=\begin{vmatrix}
    M^{i,j}&c_j^i\\
    r_i^j&\fbox{$m_{i,j}$}
    \end{vmatrix}.
\]
Note that if the entries in $M$ commute then
\begin{equation}\label{commute}
|M|_{i,j}=(-1)^{i+j}\frac{\det(M)}{\det(M^{i,j})}.
\end{equation}

\paragraph{Invariance under row and column operations}
The quasideterminants of a matrix have invariance properties
similar to those of determinants under elementary row and column
operations applied to the matrix. Consider the following
quasideterminant of an $n\times n $ matrix:
\begin{equation}\label{invariance}
    \begin{vmatrix}
    \begin{pmatrix}
      E&0\\
      F&g
    \end{pmatrix}
    \begin{pmatrix}
      A&B\\
      C&d
    \end{pmatrix}
    \end{vmatrix}_{n,n}=
    \begin{vmatrix}
      EA&EB\\
      FA+gC&FB+gd
    \end{vmatrix}_{n,n}=g(d-CA^{-1}B)=g
    \begin{vmatrix}
      A&B\\
      C&d
    \end{vmatrix}_{n,n}.
\end{equation}
The above formula can be used to understand the effect on a
quasideterminant of certain elementary row operations involving
multiplication on the left. This formula excludes those operations
which add left-multiples of the row containing the expansion point
to other rows since there is no simple way to describe the effect
of these operations. For the allowed operations however, the
results can be easily described; left-multiplying the row
containing the expansion point by $g$ has the effect of
left-multiplying the quasideterminant by $g$ and all other
operations leave the quasideterminant unchanged. There is
analogous invariance under column operations involving
multiplication on the right.

\paragraph{Noncommutative Jacobi Identity}

There is a quasideterminant version of Jacobi identity for determinants, called the noncommutative Sylvester's Theorem
by Gelfand and Retakh \cite{GR}. The simplest version of this identity is given by
\begin{equation}\label{nc syl}
    \begin{vmatrix}
      A&B&C\\
      D&f&g\\
      E&h&\fbox{$i$}
    \end{vmatrix}=
    \begin{vmatrix}
      A&C\\
      E&\fbox{$i$}
    \end{vmatrix}-
    \begin{vmatrix}
      A&B\\
      E&\fbox{$h$}
    \end{vmatrix}
    \begin{vmatrix}
      A&B\\
      D&\fbox{$f$}
    \end{vmatrix}^{-1}
    \begin{vmatrix}
      A&C\\
      D&\fbox{$g$}
    \end{vmatrix},
\end{equation}
where $f,g,h,i\in\mathcal R$, $A$ is an $n\times n$ matrix and $B,C$ (resp. $D,E$) are column (resp. row) $n$-vectors over $\mathcal R$.

\paragraph{Quasi-Pl\"{u}cker coordinates} Given an $(n+k)\times n$
matrix $A$, denote the $i$th row of $A$ by $A_i$, the submatrix of
$A$ having rows with indices in a subset $I$ of $\{1,2,\dots,n+k\}$
by $A_I$ and $A_{\{1,\dots,n+k\}\backslash\{i\}}$ by
$A_{\hat\imath}$. Given $i,j\in\{1,2,\dots,n+k\}$ and $I$ such that
$\#I=n-1$ and $j\notin I$, one defines the \emph{(right)
quasi-Pl\"{u}cker coordinates}
\begin{equation}\label{rplucker}
    r^I_{ij}=r^I_{ij}(A):=
    \begin{vmatrix}
    A_I\\
    A_i
    \end{vmatrix}_{ns}
    \begin{vmatrix}
    A_I\\
    A_j
    \end{vmatrix}_{ns}^{-1}=-
    \begin{vmatrix}
    A_I&0\\
    A_i&\fbox{0}\\
    A_j&1
    \end{vmatrix},
\end{equation}
for any column index $s\in\{1,\dots,n\}$. The final equality in
\eqref{rplucker} comes from an identity of the form \eqref{nc syl}
and proves that the definition is independent of the choice of
$s$.

\paragraph{Solving linear systems} Solutions of systems of linear systems over an arbitrary ring can be expressed in terms of quasideterminants.
\begin{thm}\label{1}
Let $A=(a_{ij})$ be an $n\times n$ matrix over a ring $\mathcal R$.
Assume that all the quasideterminants $|A|_{ij}$ are defined and
invertible. Then the system of equations
\begin{align}
x_1a_{1i}+x_2a_{2i}+\cdots+x_na_{ni}=b_i,\,\ 1\le i\le n
\end{align}
has the unique solution
\begin{align}
x_i=\sum_{j=1}^nb_j|A|_{ij}^{-1}, \,\ i=1,\ldots,n.
\end{align}
\end{thm}
Let $A_l(b)$ be the $n\times n$ matrix obtained by replacing the
$l$-th row of the matrix $A$ with the row $(b_1,\ldots,b_n)$. Then
we have the following noncommutative version of \emph{Cramer's rule}.
\begin{thm}\label{2}
In notation of Theorem~\ref{1}, if the quasideterminants $|A|_{ij}$ and
$|A_i(b)|_{ij}$ are well defined, then
\[x_i|A|_{ij}=|A_i(b)|_{ij}.\]
\end{thm}

\subsection{Relationship between quasideterminants and superdeterminants}
The basic formulae connecting quasideterminants of even supermatrices with their Berezinians are given in \cite{BMRJ}.
\begin{thm}
Let $\mathcal{M}$ be an $(m|n)\times(m|n)$-supermatrix. Then
\begin{equation}\label{Ber}
    |\mathcal{M}|_{i,j}=
    \begin{cases}
        (-1)^{i+j}\dfrac{\Ber(\mathcal{M})}{\Ber(\mathcal{M}^{i,j})}& 1\le i,j\le m,\\[12pt]
        (-1)^{i+j}\dfrac{\Berstar(\mathcal{M})}{\Berstar(\mathcal{M}^{i,j})}& m+1\le i,j\le m+n,
    \end{cases}
\end{equation}
\emph{(}cf. \eqref{commute}.\emph{)}
\end{thm}
Roughly speaking, a quasideterminant with indices in one of the even blocks of $\mathcal M$ is given as a ratio of Berezinians. A quasideterminant with its indices in the one of the odd blocks is not well-defined.

\section{Darboux transformations for twisted derivations}
Consider now a more general setting in which $\sA$ is an associative, unital algebra over ring $K$, not necessarily graded. Suppose that there is a homomorphism $\sigma\colon\sA\to\sA$ (i.e.\ for all $\alpha\in K$, $a,b\in\sA$, $\sigma(\alpha a)=\alpha \sigma(a)$, $\sigma(a+b)=\sigma(a)+\sigma(b)$ and $\sigma(ab)=\sigma(a)\sigma(b)$) and a \emph{twisted derivation} or $\sigma$-derivation \cite{HLS,CDC} $D\colon\sA\to\sA$ satisfying $D(K)=0$ and $D(ab)=D(a)b+\sigma(a)D(b)$.

Simple examples arise in the case that elements $a\in\sA$ depend on a variable $x$, say.
\begin{description}
	\item[Derivative] Here $D=\partial/\partial x$ satisfies $D(ab)=D(a)b+aD(b)$ and $\sigma$ is the identity mapping.
	\item[Forward difference] The homomorphism is the shift operator $T$, where $T(a(x))=a(x+1)$ and the twisted derivation is
	\[
	\Delta(a(x))=\frac{a(x+h)-a(x)}h,
	\]
satisfying $\Delta(ab)=D(a)b+T(a)D(b)$.
	\item[Jackson derivative] The homomorphism is a $q$-shift operator defined by $S_q(a(x))=a(qx)$ and the twisted derivation is
	\[
	D_q(a(x))=\frac{a(qx)-a(x)}{(q-1)x}.
	\]
satisfying $D_q(ab)=D_q(a)b+S_q(a)D_q(b)$.
	\item[Superderivative] As described in Section~\ref{sec:super}, for $a,b\in\sA$, a superalgebra, $D(ab)=D(a)b+\hat{a}D(b)$ where  $\hat\ $ is the grade involution.
\end{description}

There are a number of simple properties of such a twisted derivation, which are summarized in the following lemma.
\begin{lemma}\label{lem:basic properties}
\begin{enumerate}
\item Let $A, B$ be matrices over $\sA$. Whenever $AB$ is defined, $\sigma(AB)=\sigma(A)\sigma(B)$ and $D(AB)=D(A)B+\sigma(A)D(B)$,
\item Let $A$ be an invertible matrix over $\sA$. Then $\sigma(A)^{-1}=\sigma(A^{-1})$ and $D(A^{-1})=-\sigma(A)^{-1}D(A)A^{-1}$,
\item Let $A,B,C$ be matrices over $\sA$ such that $AB^{-1}C$ is well-defined. Then
\[
    D(AB^{-1}C)=D(A)B^{-1}C+\sigma(A)\sigma(B)^{-1}(D(C)-D(B)B^{-1}C).
\]
\end{enumerate}
\end{lemma}

\subsection{Darboux transformations}\label{sec:DT}
Let $\theta_0,\theta_1,\theta_2,\ldots$ be a sequence in $\sA$. Consider the sequence $\theta[0],\theta[1],\theta[2],\ldots$ in $\sA$, generated from the first sequence by Darboux transformations of the form
\begin{equation}
	G_\theta=\sigma(\theta) D\theta^{-1}=D-D(\theta)\theta^{-1},
\end{equation}
where $D$ and $\sigma$ are the twisted derivation and homomorphism defined above. To be specific, $\theta[0]=\theta_0$ and $G[0]=G_{\theta[0]}$, then let
\begin{equation}\label{theta[1]}
	\theta[1]=G[0](\theta_1)=D(\theta_1)-D(\theta_0)\theta_0^{-1}\theta_1
\end{equation}
and $G[1]=G_{\theta[1]}$, $\theta[2]=G[1]\circ G[0](\theta_2)$ and $G[2]=G_{\theta[2]}$ and so on. In general, for $k\in\mathbb{N}$,
\begin{equation}\label{G[k]}
\theta[k]=G[k-1]\circ G[k-2]\circ\dots\circ G[0](\theta_k),\quad G[k]=\sigma(\theta[k])D\theta[k]^{-1},	
\end{equation}
and we require that each $\theta[k]$ is invertible.

In the standard case, $D=\partial$ and $\sigma=\mathrm{Id}$, it is well known that the terms in the sequence of Darboux transformations have closed form expressions in terms of the original sequence. In the case that $\sA$ is commutative, they are expressed as ratios of wronskian determinants \cite{Crum},
\begin{equation}
\theta[n]=\frac{
\begin{vmatrix}
\theta_0&\dots&\theta_{n-1}&\theta_n\\
\theta^{(1)}_0&\dots&\theta^{(1)}_{n-1}&\theta^{(1)}_n\\
\vdots&&\vdots\\
\theta^{(n-1)}_0&\dots&\theta^{(n-1)}_{n-1}&\theta^{(n-1)}_{n}\\
\theta^{(n)}_0&\dots&\theta^{(n)}_{n-1}&\theta^{(n)}_n
\end{vmatrix}
}{
\begin{vmatrix}
\theta_0&\dots&\theta_{n-1}\\
\theta^{(1)}_0&\dots&\theta^{(1)}_{n-1}\\
\vdots&&\vdots\\
\theta^{(n-1)}_0&\dots&\theta^{(n-1)}_{n-1}
\end{vmatrix}
},\quad n\in\mathbb{N},	
\end{equation}
where $\theta^{(i)}_j$ denotes $\partial^i(\theta_j)$. In the case that $\sA$ is not commutative, the terms in the sequence are expressed as quasideterminants \cite{EGR},
\begin{equation}
\theta[n]=
\begin{vmatrix}
\theta_0&\dots&\theta_{n-1}&\theta_n\\
\theta^{(1)}_0&\dots&\theta^{(1)}_{n-1}&\theta^{(1)}_n\\
\vdots&&\vdots\\
\theta^{(n-1)}_0&\dots&\theta^{(n-1)}_{n-1}&\theta^{(n-1)}_{n}\\
\theta^{(n)}_0&\dots&\theta^{(n)}_{n-1}&\fbox{$\theta^{(n)}_n$}
\end{vmatrix},\quad n\in\mathbb{N}.
\end{equation}
The following theorem gives a generalisation of this formula to the case of general $D$ and $\sigma$. Note in particular that the expressions do not depend on $\sigma$ and are obtained simply by replacing $\partial$ with $D$.
\begin{thm}\label{thm:DT}
Let $\phi[0]=\phi$ and for $n\in\mathbb{N}$ let
\begin{equation*}
\phi[n]=D(\phi[n-1])-D(\theta[n-1])\theta[n-1]^{-1}\phi[n-1],	
\end{equation*}
where $\theta[n]=\phi[n]|_{\phi\to\theta_n}$.
Then, for $n\in\mathbb{N}$,
\begin{equation}\label{theta[n]}
	\phi[n]=
	\begin{vmatrix}
	    \theta_0&\cdots&\theta_{n-1}&\phi\\
	    D(\theta_0)&\cdots&D(\theta_{n-1})&D(\phi)\\
	    \vdots&&\vdots&\vdots\\
	    D^{n-1}(\theta_0)&\cdots&D^{n-1}(\theta_{n-1})&D^{n-1}(\phi)\\
	    D^{n}(\theta_0)&\cdots&D^{n}(\theta_{n-1})&\fbox{$D^{n}(\phi)$}
	\end{vmatrix}.
\end{equation}
\end{thm}
\begin{proof}
(By induction.) The case $n=1$ follows from the definition \eqref{defn} and \eqref{theta[1]}. Now assume that \eqref{theta[n]} holds for $n=k$. Then
\[
\phi[k]=D^k(\phi)-
\begin{bmatrix}
    D^k(\theta_0)&\cdots& D^k(\theta_{k-1})
\end{bmatrix}
\Theta^{-1}
\begin{bmatrix}
    \phi\\D(\phi)\\\vdots\\D^{k-1}(\phi)
\end{bmatrix},
\]
where
\[
    \Theta=\begin{bmatrix}
        \theta_0&\cdots&\theta_{k-1}\\
        D(\theta_0)&\cdots&D(\theta_{k-1})\\
        \vdots&&\vdots\\
        D^{k-1}(\theta_0)&\cdots&D^{k-1}(\theta_{k-1})
    \end{bmatrix}.
\]
To complete the proof we must show that $\phi[k+1]=D(\phi[k])-D(\theta[k])\theta[k]^{-1}\phi[k]$ can be written in the form \eqref{theta[n]} for $n=k+1$.

Using Lemma~\ref{lem:basic properties}, one obtains
\[
    D(\phi[k])=
    \begin{vmatrix}
        \theta_0&\cdots&\theta_{k-1}&\phi\\
        D(\theta_0)&\cdots&D(\theta_{k-1})&D(\phi)\\
        \vdots&&\vdots&\vdots\\
        D^{k-2}(\theta_0)&\cdots&D^{k-2}(\theta_{k-1})&D^{k-2}(\phi)\\
        D^{k-1}(\theta_0)&\cdots&D^{k-1}(\theta_{k-1})&D^{k-1}(\phi)\\
        D^{k+1}(\theta_0)&\cdots&D^{k+1}(\theta_{k-1})&\fbox{$D^{k+1}(\phi)$}
    \end{vmatrix}+
    \sigma\left(\begin{vmatrix}
        \theta_0&\cdots&\theta_{k-1}&0\\
        D(\theta_0)&\cdots&D(\theta_{k-1})&0\\
        \vdots&&\vdots&\vdots\\
        D^{k-2}(\theta_0)&\cdots&D^{k-2}(\theta_{k-1})&0\\
        D^{k-1}(\theta_0)&\cdots&D^{k-1}(\theta_{k-1})&1\\
        D^{k}(\theta_0)&\cdots&D^{k}(\theta_{k-1})&\fbox{$0$}
    \end{vmatrix}\right)
    \phi[k],
\]
and so,
\[
    D(\theta[k])=
    \begin{vmatrix}
        \theta_0&\cdots&\theta_{k-1}&\theta_{k}\\
        D(\theta_0)&\cdots&D(\theta_{k-1})&D(\theta_{k})\\
        \vdots&&\vdots&\vdots\\
        D^{k-2}(\theta_0)&\cdots&D^{k-2}(\theta_{k-1})&D^{k-2}(\theta_{k})\\
        D^{k-1}(\theta_0)&\cdots&D^{k-1}(\theta_{k-1})&D^{k-1}(\theta_{k})\\
        D^{k+1}(\theta_0)&\cdots&D^{k+1}(\theta_{k-1})&\fbox{$D^{k+1}(\theta_{k})$}
    \end{vmatrix}+
    \sigma\left(\begin{vmatrix}
        \theta_0&\cdots&\theta_{k-1}&0\\
        D(\theta_0)&\cdots&D(\theta_{k-1})&0\\
        \vdots&&\vdots&\vdots\\
        D^{k-2}(\theta_0)&\cdots&D^{k-2}(\theta_{k-1})&0\\
        D^{k-1}(\theta_0)&\cdots&D^{k-1}(\theta_{k-1})&1\\
        D^{k}(\theta_0)&\cdots&D^{k}(\theta_{k-1})&\fbox{$0$}
    \end{vmatrix}\right)
    \theta[k].
\]
Then, using \eqref{nc syl}, $D(\phi[k])-D(\theta[k])\theta[k]^{-1}\phi[k]$ can be expressed in the required form.
\end{proof}
As an application of this theorem, in the following section we will use it to construct solutions of the super KdV equation.

\section{The Manin-Radul Super KdV equation}
The Manin-Radul supersymmetric KdV (MRSKdV) system \cite{ManinRadul} is
\begin{align}
\label{MRSKdV}
\alpha_t=\tfrac14(\alpha_{xx}+3\alpha D(\alpha)+6\alpha u)_x,\quad u_t=\tfrac14(u_{xx}+3u^2+3\alpha D(u))_x,
\end{align}
where $u$ and $\alpha$ are even and odd dependent variables respectively, $x,t$ are even independent variables variables and $D$ is the superderivative defined by $D=\partial_\theta+\theta\partial_x$, where $\theta$ is a Grassmann odd variable, satisfying $D^2=\partial_x$. This system has the Lax pair
\begin{align}
\label{L}
	L&=\partial_x^2+\alpha D+u,\\	
\label{M}
	M&=\partial_x^3+\tfrac34\bigl((\alpha\partial_x+\partial_x\alpha)D+u\partial_x+\partial_xu),	
\end{align}
in the sense that $L_t+[L,M]=0$ implies \eqref{MRSKdV}. Eigenfunctions satisfy
\begin{equation}
\label{LP}
L(\phi)=\lambda\phi,\quad\phi_t=M(\phi),
\end{equation}
for eigenvalue $\lambda$.

\subsection{Darboux transformations}
A Darboux transformation for this system \cite{QPM} is
\begin{align}
\phi&\to D(\phi)-D(\theta)\theta^{-1}\phi,\\
\alpha&\to-\alpha+2(D(\theta)\theta^{-1})_x,\\
u&\to u+D(\alpha)-2D(\theta)\theta^{-1}(\alpha-(D(\theta)\theta^{-1})_x),
\end{align}
where $\theta$ is an invertible, and hence necessarily even, solution of \eqref{LP}. Note that it is an example of the general type of Darboux transformation discussed in Section~\ref{sec:DT}. As discussed there, this transformation may be iterated by taking solutions $\theta_0,\theta_1,\theta_2,\dots$ of \eqref{LP} to obtain
\begin{align}
\phi[k+1]&=D(\phi[k])-D(\theta[k])\theta[k]^{-1}\phi[k],\\
\theta[k]&=\phi[k]|_{\phi\to\theta_k}.
\end{align}
The requirement that each $\theta[k]$ is invertible means that it must be even and consequently that $\theta_i$ must have parity $i$. The corresponding solutions of MRSKdV are $\alpha[0]=\alpha$, $u[0]=u$ and
\begin{align}
\label{alpha rec}
\alpha[k+1]&=-\alpha[k]+2(D(\theta[k])\theta[k]^{-1})_x,\\
\label{u rec}
u[k+1]&=u[k]+D(\alpha[k])-2D(\theta[k])\theta[k]^{-1}(\alpha[k]-(D(\theta[k])\theta[k]^{-1})_x).
\end{align}

Theorem~\ref{thm:DT} give a closed-form expression \eqref{theta[n]} for $\phi[n]$ as a quasideterminant. We will obtain the  corresponding expressions for $\alpha[n]$ and $u[n]$ in terms quasideterminants of the form
\begin{align}
\label{Qn orig}
Q_n(i,j)&=
\begin{vmatrix}
\theta_0&\cdots&\theta_{n-1}&0\\
D(\theta_0)&\cdots&D(\theta_{n-1})&0\\
\vdots&\ddots&\vdots&\vdots\\
D^{n-j-2}(\theta_0)&\cdots&D^{n-j-2}(\theta_{n-1})&0\\
D^{n-j-1}(\theta_0)&\cdots&D^{n-j-1}(\theta_{n-1})&1\\
D^{n-j}(\theta_0)&\cdots&D^{n-j}(\theta_{n-1})&0\\
\vdots&\ddots&\vdots&\vdots\\
D^{n-1}(\theta_0)&\cdots&D^{n-1}(\theta_{n-1})&0\\
D^{n+i}(\theta_0)&\cdots&D^{n+i}(\theta_{n-1})&\fbox{$0$}
\end{vmatrix}\\[8pt]
\label{Qn}
&=-\begin{vmatrix}
\theta_0&\cdots&\theta_{n-1}\\
\vdots&\ddots&\vdots\\
D^{n-j-2}(\theta_0)&\cdots&D^{n-j-2}(\theta_{n-1})\\
D^{n-j}(\theta_0)&\cdots&D^{n-j}(\theta_{n-1})\\
\vdots&\ddots&\vdots\\
D^{n-1}(\theta_0)&\cdots&D^{n-1}(\theta_{n-1})\\
D^{n+i}(\theta_0)&\cdots&D^{n+i}(\theta_{n-1})\\
\end{vmatrix}_{n,s}
\begin{vmatrix}
\theta_0&\cdots&\theta_{n-1}\\
\vdots&\ddots&\vdots\\
D^{n-1}(\theta_0)&\cdots&D^{n-1}(\theta_{n-1})
\end{vmatrix}_{n-j,s}^{-1},
\end{align}
for any $s=1,\ldots,n$ (see \eqref{rplucker}).

The following lemma records some useful properties of the $Q_n(i,j)$.
\begin{lemma}\label{lem:Qn}
\emph{(\emph{i})} $\hat{Q_n(i,j)}=(-1)^{i+j+1}Q_n(i,j)$. i.e., $Q_n(i,j)$ has the parity $(-1)^{i+j+1}$,\\[6pt]
\emph{(\emph{ii})} $D(\theta_0)\theta_0^{-1}=-Q_1(0,0)$ and
$D(\theta[k])\theta[k]^{-1}=-Q_k(0,0)-Q_{k+1}(0,0)$ for $k\ge 1$,\\[6pt]
\emph{(\emph{iii})} $Q_{n+1}(0,1)=Q_n(1,0)+Q_{n+1}(0,0)Q_{n}(0,0)$.
\end{lemma}
\begin{proof}
(i) This follows from the facts that $\theta_i$ has parity $i$ and $D$ is odd, and from the invariance properties of quasideterminants \eqref{invariance}.\\[6pt]	
(ii) It is obvious from the definition \eqref{Qn} that $Q_{1}(0,0)=-D(\theta_0)\theta_0^{-1}$. The second result follows by using the same method used in the proof of Theorem~\ref{thm:DT} with $\sigma=\hat{\quad}$.\\[6pt]
(iii) Expanding $Q_{n+1}(0,1)$, written in the form \eqref{Qn orig}, using \eqref{nc syl} gives the terms on the right hand side.
\end{proof}

\begin{thm}\label{thm:alpha u}
After $n$ repeated Darboux transformations, the MRSKdV system has
new solutions $\alpha[n]$ and $u[n]$ expressed in terms of $Q_n(0,0)$ and $Q_n(0,1)$.
\begin{align}
\alpha[n]&=(-1)^n\alpha-2Q_n(0,0)_x,\label{alpha s}\\
u[n]&=u-2Q_n(0,1)_x-2Q_n(0,0)((-1)^n\alpha-Q_n(0,0)_x)+{1-(-1)^n\over 2}D(\alpha).\label{u s}
\end{align}
\end{thm}
\begin{proof}
(By induction.) First consider \eqref{alpha s}. For $n=1$ we must show that
\[
\alpha[1]=-\alpha+2(D(\theta_0)\theta_0^{-1})_x.
\]
This is clear from \eqref{alpha rec} for $k=0$ and Lemma~\ref{lem:Qn}(ii).

By using \eqref{alpha rec} in general case and Lemma~\ref{lem:Qn}(ii), assuming that \eqref{alpha s} holds for $n=k$, we have
\begin{align*}
\alpha[k+1]&=-\alpha[k]+2(D(\theta[k])\theta[k]^{-1})_x\\
&=(-1)^{k+1}\alpha-2Q_{k+1}(0,0)_x,
\end{align*}
as required.

The proof of \eqref{u s} is very similar and makes use of Lemma~\ref{lem:Qn}(iii). \end{proof}

\begin{rem} The quasideterminant solutions \eqref{alpha s} and \eqref{u s} could also be obtained by solving a noncommutative
linear system. In \cite{QPM}, the solutions $\alpha[n]$ and $u[n]$ for the MRSKdV system were given as
\begin{align}
\alpha[n]&=(-1)^n\alpha-2\partial a_{n,n-1},\\ u[n]&=u-2\partial a_{n,n-2}-2a_{n,n-1}((-1)^n\alpha-\partial
a_{n,n-1})+{1-(-1)^n\over 2}D\alpha,
\end{align}
where $a_{n,n-1},a_{n,n-2},\ldots,a_{n,0}$ satisfy the linear system
\begin{align}\label{LS} (D^n+a_{n,n-1}D^{n-1}+\cdots+a_{n,0})\theta_j=0,\,\ j=0,\ldots,n-1.
\end{align}
By using Theorem~\ref{2} and \eqref{rplucker}, we can solve the above linear system to obtain
\begin{align*} a_{n,n-i}=Q_n(0,i-1),\quad
i=1,2,\ldots,n.
\end{align*}
It is obvious that solutions $\alpha[n]$ and $u[n]$ obtained here coincide with those given by \eqref{alpha s} and \eqref{u s}. \end{rem}

\subsection{From quasideterminants to superdeterminants I}
It is natural for a supersymmetric system to express the solutions in terms of superdeterminants and this was done, in all but one case in \cite{QPM}. In this section, we will show systematically that the solutions given in Theorem~\ref{thm:alpha u}, expressed in terms of quasideterminants $Q_n(0,0)$ and $Q_n(0,1)$, can be reexpressed in terms of superdeterminants. The expressions we will obtain coincide with the superdeterminant solutions found in \cite{QPM} and we also find the superdeterminant expressions in the case that they did not.

Recall the property that $\theta_i$ has parity $i$. Let us therefore introduce the relabeling
\begin{equation}
\theta_{2k}=E_{k},\quad \theta_{2k+1}=O_{k}.	
\end{equation}
Also, we write $D^{2j}(\theta)=\theta^{(j)}$ and $D^{2j+1}(\theta)=D(\theta^{(j)})$, where $^{(j)}$ denotes the $j$th derivative with respect to $x$.

Consider the matrix
\begin{equation}
	W_n=\begin{bmatrix}
	\theta_0&\cdots&\theta_{n-1}\\
	\vdots&\ddots&\vdots\\
	D^{n-1}(\theta_0)&\cdots&D^{n-1}(\theta_{n-1})
	\end{bmatrix},
\end{equation}
appearing in the definition \eqref{Qn} of $Q_n(i,j)$. There is a natural reordering of the rows and columns
\begin{equation}
	W_n\to
	\mathcal W_n=\begin{bmatrix}
	X_n&Y_n\\
	Z_n&T_n
	\end{bmatrix},
\end{equation}
which gives an even matrix $\mathcal W_n$. This reordering does not change the value of any associated quasideterminant, as long as the expansion point in each refers to the same element. In the case that $n$ is even,
\begin{align*}
X_{2k}=	
\begin{bmatrix}
	E_{0}&\cdots&E_{k-1}\\
	\vdots&\ddots&\vdots\\
	E^{(k-1)}_{0}&\cdots&E^{(k-1)}_{k-1}
\end{bmatrix},\quad
Y_{2k}=	
\begin{bmatrix}
	O_{0}&\cdots&O_{k-1}\\
	\vdots&\ddots&\vdots\\
	O^{(k-1)}_{0}&\cdots&O^{(k-1)}_{k-1}
\end{bmatrix}
\end{align*}
and $Z_{2k}=D(X_{2k})$ and $T_{2k}=D(Y_{2k})$ are all $k\times k$ matrices. In the case that $n$ is odd, $X_{2k+1}$ is $(k+1)\times(k+1)$, $Y_{2k+1}$ is $(k+1)\times k$, $Z_{2k+1}$ is $k\times(k+1)$ and $T_{2k+1}$ is a $k\times k$ matrix whose precise form can be easily deduced from the above description.

Similarly, consider the matrix
\begin{equation}
	W'_n=\begin{bmatrix}
	\theta_0&\cdots&\theta_{n-1}\\
	\vdots&\ddots&\vdots\\
	D^{n-3}(\theta_0)&\cdots&D^{n-3}(\theta_{n-1})\\
	D^{n-1}(\theta_0)&\cdots&D^{n-1}(\theta_{n-1})\\
	D^{n}(\theta_0)&\cdots&D^{n}(\theta_{n-1})
	\end{bmatrix},
\end{equation}
appearing in the definition \eqref{Qn} of $Q_n(0,1)$. A similar reordering of this matrix
\begin{equation}
	W'_n\to
	\mathcal W'_n=\begin{bmatrix}
	X'_n&Y'_n\\
	Z'_n&T'_n
	\end{bmatrix},
\end{equation}
gives another even matrix $\mathcal W'_n$ where, for example,
\begin{align*}
X_{2k}'=	
\begin{bmatrix}
	E_{0}&\cdots&E_{k-1}\\
	\vdots&\ddots&\vdots\\
	E^{(k-2)}_{0}&\cdots&E^{(k-2)}_{k-1}\\
	E^{(k)}_{0}&\cdots&E^{(k)}_{k-1}
\end{bmatrix}.
\end{align*}

\begin{thm}\label{thm:summary}
For $n\in\mathbb{N}$,
\begin{equation}
Q_{n}(0,0)=D(\log(\B(\mathcal W_n))),	\quad
Q_{n}(0,1)=-\dfrac{\B(\mathcal W_n')}{\B(\mathcal W_n)},	
\end{equation}	
where $\B=\Ber$ if $n$ is even, and $\B=\Berstar$ if $n$ is odd.
\end{thm}
\begin{proof}
Using a straightforward generalization of the method of differentiating quasideterminants given in \cite{CGJN}, one can show that
\begin{equation}\label{Qn00}
	Q_{n}(0,0)=-D(|W_n|_{n,n})|W_n|_{n,n}^{-1}-Q_{n-1}(0,0)=-D(\log|W_n|_{n,n})-Q_{n-1}(0,0).
\end{equation}
Note that $|\hat W_n|_{n,n}=|W_n|_{n,n}$ and so it is even.

If $n=2k$, then by \eqref{Ber}
\[
|W_n|_{n,n}=|\mathcal W_{2k}|_{2k,2k}=\frac{\Berstar(\mathcal W_n)}{\Berstar(\mathcal W_{n-1})}=\frac{\Ber(\mathcal W_{n-1})}{\Ber(\mathcal W_{n})}.
\]	
Similarly, if $n=2k+1$, then
\[
|W_n|_{n,n}=|\mathcal W_{2k+1}|_{k+1,k+1}=\frac{\Ber(\mathcal W_{n})}{\Ber(\mathcal W_{n-1})},
\]
By repeated use of \eqref{Qn00}, noting that $Q_1(0,0)=-D(\log\theta_0)$, it follows that when $n$ is even
\[
Q_{n}(0,0)=D(\log(\Ber(\mathcal W_n)),
\]
and when $n$ is odd
\[
Q_{n}(0,0)=D(\log(\Berstar(\mathcal W_n)),
\]
as required.

Also, from \eqref{Qn}, we have
\begin{equation}
	Q_n(0,1)=-|W_n'|_{n,n-1}||W_n|^{-1}_{n-1,n-1}=-
	\begin{cases}
		\dfrac{|\mathcal W'_{2k}|_{k,k}}{|\mathcal W_{2k}|_{k,k}}&n=2k\\
		\dfrac{|\mathcal W'_{2k+1}|_{2k+1,2k+1}}{|\mathcal W_{2k+1}|_{2k+1,2k+1}}&n=2k+1
	\end{cases}.	
\end{equation}
The two quasideterminant in these expressions can be expressed, by \eqref{Ber}, as a ratio of Berezinians in which a common factor cancels. The result is that
\begin{equation}
	Q_n(0,1)=-
	\begin{cases}
		\dfrac{\Ber(\mathcal W'_{2k})}{\Ber(\mathcal W_{2k})}&n=2k\\[8pt]
		\dfrac{\Berstar(\mathcal W'_{2k+1})}{\Berstar(\mathcal W_{2k+1})}&n=2k+1
	\end{cases},	
\end{equation}
as required.
\end{proof}

\subsection{Binary Darboux transformations}
Binary Darboux transformations for the MRSKdV system were discussed in \cite{ShTu,LM}. In these articles, solutions expressed in terms of determinants were obtained. As discussed in connection with Darboux transformations it is to be expected that solutions for this supersymmetric system should be \emph{superdeterminants} in general. In this section, we will construct a more general type of binary Darboux transformation which will be shown to give these superdeterminants solutions and includes the solutions found in \cite{ShTu,LM} as a special case.

First we recall the definition of the adjoint for supersymmetric linear operators. For a linear operator $P$, $|P|$ denotes its parity. For example, $|D|=1$ and $|\partial|=0$, where $\partial$ denotes any derivative with respect to an even variable, and the parity of multiplication by a homogeneous element is the parity of that element (in the usual sense). The rules defining the superadjoint are
\begin{align}
	D^\dagger=-D,\quad\partial^\dagger=-\partial,\quad\mathcal M^\dagger=\mathcal M^{st},
\end{align}
where $\mathcal M$ denotes any matrix over $\sA$, together with the product rule
\begin{align}
	(PQ)^\dagger=(-1)^{|P||Q|}Q^\dagger P^\dagger,
\end{align}
where $P$ and $Q$ are operators (cf.~\eqref{st} for the case of matrices). In particular, this gives $(D^n)^\dagger=(-1)^{n(n+1)/2}D^n$ and, consistently, $(\partial^n)^\dagger=(-1)^n\partial^n$. For any $a\in\sA$, $a^\dagger=a$.

The Lax pair \eqref{L}, \eqref{M} has the adjoint form
\begin{align}
\label{Lad}
	L^\dagger&=\partial_x^2+D\alpha+u,\\	
\label{Mad}
	M^\dagger&=-\partial_x^3-\tfrac34\bigl(D(\alpha\partial_x+\partial_x\alpha)+u\partial_x+\partial_xu),	
\end{align}
and adjoint eigenfunctions satisfy
\begin{equation}
\label{ALP}
L^\dagger(\psi)=\xi\psi,\quad-\psi_t=M^\dagger(\psi),
\end{equation}
for eigenvalue $\xi$. Given an (eigenfunction, adjoint) eigenfunction pair $(\theta,\rho)$, the binary Darboux transformation \cite{ShTu,LM} is given by
\begin{align}
\label{bDT1}
	\phi&\to\phi-\theta\Omega(\theta,\rho)^{-1}\Omega(\phi,\rho),\\
\label{bDT2}
	\psi&\to\psi-\rho\Omega(\theta,\rho)^{-1}\Omega(\theta,\psi),\\
	\label{bDTa}
	\alpha&\to\alpha+2(\theta\Omega(\theta,\rho)^{-1}\hat{\rho})_x\\
	\label{bDTu}
	u&\to u-2(\alpha+(\theta\Omega(\theta,\rho)^{-1}\hat{\rho})_x)\theta\Omega(\theta,\rho)^{-1}\hat{\rho}+2(\theta\Omega(\theta,\rho)^{-1}D(\rho))_x,
\end{align}
where eigenfunction $\theta$ and adjoint eigenfunction $\rho$ have opposite parities. Since $D\bigl(\Omega(\phi,\psi\bigr)=\psi\phi$, $\Omega$ is even and assumed to be invertible. When iterating this transformation, both previous papers \cite{ShTu,LM} on this topic considered the case that all eigenfunctions are even and all adjoint eigenfunctions are odd. We will show that this is not the most general possibility however.

Consider an even $(m|n)$-row vector eigenfunction $\mathcal{E}=(\theta_{0},\dots\theta_{m+n-1})$ and an odd $(m|n)$-row vector adjoint eigenfunction $\mathcal{O}=(\rho_{0},\dots,\rho_{m+n-1})$, where $\theta_{i}$ for $i=0,\dots,m-1$ and $\rho_{m+j}$ for $j=0,\dots,n-1$ are even and $\rho_{i}$ for $i=0,\dots,m-1$ and $\theta_{m+j}$ for $j=0,\dots,n-1$ are odd. These row vectors satisfy
\begin{gather}
\label{E}
	L(\mathcal E)=\mathcal E\Lambda,\quad\mathcal E_t=M(\mathcal E),\\
\label{O}
	L^\dagger(\mathcal O)=\mathcal O\Xi,\quad-\mathcal O_t=M^\dagger(\mathcal O),
\end{gather}
where $\Lambda$ and $\Xi$ are constant $(m+n)\times(m+n)$ diagonal matrices containing the eigenvalues. Then $\Omega=\Omega(\mathcal{E},\mathcal{O})$ is an even $(m|n)\times(m|n)$-supermatrix defined up to a constant by
\begin{align}
D(\Omega)&=\mathcal{O}^\dagger\mathcal{E},\qquad
\Omega\Lambda-\Xi\Omega=D(\mathcal{O}^\dagger\mathcal{E}_x-\mathcal{O}_x^\dagger\mathcal{E})-\hat{\mathcal{O}}^\dagger\alpha\mathcal{E},\\
\Omega_t&=D(\mathcal{O}_{xx}^\dagger\mathcal{E}-\mathcal{O}_x^\dagger\mathcal{E}_x+\mathcal{O}^\dagger\mathcal{E}_{xx})+\tfrac32\hat{\mathcal{O}}_x^\dagger\alpha\mathcal{E}+\tfrac34\hat{\mathcal{O}}^\dagger\alpha_x\mathcal{E}+\tfrac32D\bigl(\mathcal{O}^\dagger\alpha D(\mathcal{E})\bigr)+\tfrac32D\bigl(\mathcal{O}^\dagger u\mathcal{E}\bigr).
\end{align}

The closed form expressions for the results of iterated binary Darboux transformations are stated in the following theorems.
\begin{thm}\label{thm:Phi Psi}
Iterating the binary Darboux transformations \eqref{bDT1}, \eqref{bDT2} for $m+n\ge1$, one obtains	
\begin{equation}
\phi[m+n]=\begin{vmatrix}
\Omega(\mathcal E,\mathcal O)&\Omega(\phi,\mathcal O)\\
\mathcal E&\fbox{$\phi$}
\end{vmatrix},\quad
\psi[m+n]=\begin{vmatrix}
\Omega(\mathcal E,\mathcal O)^\dagger&\Omega(\mathcal E,\psi)^\dagger\\
\mathcal O&\fbox{$\psi$}
\end{vmatrix}
=\begin{vmatrix}
\hat{\Omega(\mathcal E,\mathcal O)}&\mathcal O^\dagger\\
\hat{\Omega(\mathcal E,\psi)}&\fbox{$\psi$}
\end{vmatrix}\label{Phi Psi}
\end{equation}
with
\begin{equation}\label{DQ}
\Omega(\phi[m+n],\psi[m+n])=
\begin{vmatrix}
    \Omega(\mathcal E,\mathcal O)&\Omega(\phi,\mathcal O)\\
    \Omega(\mathcal E,\psi)&\fbox{$\Omega(\phi,\psi)$}
\end{vmatrix}.
\end{equation}
\end{thm}

\begin{proof} (By induction.) Let $k=m+n$. The formulae \eqref{Phi Psi} and \eqref{DQ} evidently hold for $k=1$. Now suppose that they hold for given $k$ and consider another binary Darboux transformation defined in terms of
\begin{equation}
\theta[k]=\begin{vmatrix}
\Omega(\mathcal E,\mathcal O)&\Omega(\theta,\mathcal O)\\
\mathcal E&\fbox{$\theta$}
\end{vmatrix},\quad
\rho[k]=\begin{vmatrix}
\hat{\Omega(\mathcal E,\mathcal O)}&\mathcal O^\dagger\\
\hat{\Omega(\mathcal E,\rho)}&\fbox{$\rho$}
\end{vmatrix}\label{theta rho}.
\end{equation}
Here $(\theta,\rho)$ be another eigenfunction, adjoint eigenfunction pair of opposite parity, so that $(\mathcal E,\theta)$ and $(\mathcal O,\rho)$ are even and odd row vectors respectively. Then, using \eqref{nc syl},
\begin{align*}
	\phi[k+1]&=\phi[k]-\theta[k]\Omega(\theta[k],\rho[k])^{-1}\Omega(\phi[k],\rho[k])\\
	&=\begin{vmatrix}
	\Omega(\mathcal E,\mathcal O)&\Omega(\theta,\mathcal O)&\Omega(\phi,\mathcal O)\\
	\Omega(\mathcal E,\rho)&\Omega(\theta,\rho)&\Omega(\phi,\rho)\\
	\mathcal E&\theta&\fbox{$\rho$}
	\end{vmatrix},
\end{align*}
verifying the first equation in \eqref{Phi Psi} for $k+1$. The second equation is proved in a similar way. The alternative expression for $\psi[m+n]$ follows from \eqref{st st}.

Finally, it is straightforward to show that $D$ applied to the RHS of \eqref{DQ} gives
\begin{align}\label{RHS}
(\psi-\hat{\Omega(\mathcal{E},\psi)}\hat{\Omega(\mathcal{E},\mathcal{O})}{}^{-1}\mathcal{O}^\dagger)(\phi-\mathcal{E}\Omega(\mathcal{E},\mathcal{O})^{-1}\Omega(\phi,\mathcal{O}))=\psi[k]\phi[k],\end{align}
as required.
\end{proof}

\begin{thm}\label{thm:alpha U}
Let $(\alpha,u)$ be a solution of MRSKdV and let $\mathcal E$ and $\mathcal O$ respectively be even and odd $(m|n)$-row vectors satisfying \eqref{E} and \eqref{O}. Then for any integers $m+n\ge0$
\begin{equation}\label{QGA}
\alpha[m+n]=\alpha-2A[m+n]_x,\quad u[m+n]=u+2(\alpha-A[m+n]_x)A[m+n]-2U[m+n]_x,
\end{equation}
where
\begin{equation}\label{AU}
	A[m+n]=
	\begin{vmatrix}
	    \Omega(\mathcal E,\mathcal O)&\hat{\mathcal O}^\dagger\\
	    \mathcal E&\fbox{$0$}
	\end{vmatrix},\quad U[m+n]=
	\begin{vmatrix}
	    \Omega(\mathcal E,\mathcal O)&D(\mathcal O^\dagger)\\
	    \mathcal E&\fbox{$0$}
	\end{vmatrix},
\end{equation}
are also solutions of MRSKdV.
\end{thm}

\begin{proof}
Let $k=m+n$. The formulae \eqref{AU} are true by definition for $k=0$; $\alpha[0]=\alpha$ and $u[0]=u$. The formulae will be proved by finding certain expression involving $\alpha[k]$ and $u[k]$ that are independent of $k$.

As in the proof of Theorem~\ref{thm:Phi Psi}, suppose that a binary Darboux transformation is determined by
\begin{equation*}
\theta[k]=\begin{vmatrix}
\Omega(\mathcal E,\mathcal O)&\Omega(\theta,\mathcal O)\\
\mathcal E&\fbox{$\theta$}
\end{vmatrix},\quad
\rho[k]=\begin{vmatrix}
\hat{\Omega(\mathcal E,\mathcal O)}&\mathcal O^\dagger\\
\hat{\Omega(\mathcal E,\rho)}&\fbox{$\rho$}
\end{vmatrix},
\end{equation*}
where $(\theta,\rho)$ are an eigenfunction, adjoint eigenfunction pair of opposite parity, so that $(\mathcal E,\theta)$ and $(\mathcal O,\rho)$ are even and odd row vectors respectively. Then, by \eqref{bDTa} and \eqref{bDTu},
\begin{align}\label{alpha[k+1]}
\alpha[k+1]&=\alpha[k]+2(\theta[k]\Omega(\theta[k],\rho[k])^{-1}\hat{\rho[k]})_x,\\
\label{u[k+1]}
u[k+1]&=u[k]-2(\alpha[k]+(\theta[k]\Omega(\theta[k],\rho[k])^{-1}\hat{\rho[k]})_x)\theta[k]\Omega(\theta[k],\rho[k])^{-1}\hat{\rho[k]}\nonumber\\&\qquad+2(\theta[k]\Omega(\theta[k],\rho[k])^{-1}D(\rho[k]))_x.
\end{align}
Using the notation introduced in \eqref{AU}, we have
\begin{equation}
	A[k+1]=
	\begin{vmatrix}
	    \Omega(\mathcal E,\mathcal O)&\Omega(\theta,\mathcal O)&\hat{\mathcal O}^\dagger\\
	    \Omega(\mathcal E,\rho)&\Omega(\theta,\rho)&\hat{\rho}\\
	    \mathcal E&\theta&\fbox{$0$}
	\end{vmatrix},\quad U[k+1]=
	\begin{vmatrix}
	    \Omega(\mathcal E,\mathcal O)&\Omega(\theta,\mathcal O)&D(\mathcal O^\dagger)\\
	    \Omega(\mathcal E,\rho)&\Omega(\theta,\rho)&D(\rho)\\
	    \mathcal E&\theta&\fbox{$0$}
	\end{vmatrix}.
\end{equation}
Then using \eqref{nc syl}, and noting that $A[k]$ and $A[k+1]$ are odd, it is straightforward to show that
\begin{equation}\label{A[k]}
\theta[k]\Omega(\theta[k],\rho[k])^{-1}\hat{\rho[k]}=A[k]-A[k+1],	
\end{equation}
and
\begin{equation}\label{U[k]}
\theta[k]\Omega(\theta[k],\rho[k])^{-1}D(\rho[k])=U[k]-U[k+1]-A[k+1]A[k].	
\end{equation}
Substituting these in \eqref{alpha[k+1]} and \eqref{u[k+1]} one obtains $k$ independent expressions
\begin{equation}
	\alpha[k+1]+2A[k+1]_x=\alpha[k]+2A[k]_x=\alpha,
\end{equation}
and
\begin{equation}
	u[k+1]-2(\alpha-A[k+1]_x)A[k+1]+2U[k+1]_x=u[k]-2(\alpha -A[k]_x)A[k]+2U[k]_x=u,
\end{equation}
and hence \eqref{QGA} hold for all $k$.
\end{proof}

\subsection{From quasideterminants to superdeterminants II}
In this section we will show how the quasideterminant solutions $(A[m+n], U[m+n])$ obtained using binary Darboux transformations can be expressed in terms of superdeterminants. To do this, it is necessary to introduce a more detailed notation for row vector eigenfunctions and adjoint eigenfunctions. Recall that for the general transformation we use $(m|n)$-row vectors $\mathcal E$ and $\mathcal O$ which are even and odd with entries $\theta_i$ and $\rho_i$ respectively. Here we will also write $\mathcal E^i=(\theta_0,\dots,\theta_{i-1})$ and $\mathcal O^i=(\rho_0,\dots,\rho_{i-1})$ for the row vectors containing the first $i$ entries of $\mathcal E$ and $\mathcal O$ respectively, and denote by subscript $0$ and $1$ the even and odd element parts of $\mathcal E$ and $\mathcal O$ respectively. Thus $\mathcal E=(\mathcal E_0,\mathcal E_1)$ and $\mathcal O=(\mathcal O_1,\mathcal O_0)$.
\begin{thm}
The expressions $(A[m+n], U[m+n])$ can expressed as
\begin{equation}\label{A U}
	A[m+n]=D\left(\log \Ber(\mathcal G_{(m|n)})\right),\quad
	U[m+n]=\frac{\Ber(\mathcal G'_{(m+1|n)})}{\Ber(\mathcal G_{(m|n)})},
\end{equation}
where
\[
\mathcal G_{(m|n)}=\begin{pmat}({|})
\Omega(\mathcal E_0,\mathcal O_1)&\Omega(\mathcal E_1,\mathcal O_1)\cr\-
\Omega(\mathcal E_0,\mathcal O_0)&\Omega(\mathcal E_1,\mathcal O_0)\cr
\end{pmat},
\]
in an even $(m|n)\times(m|n)$-supermatrix and
\[	
\mathcal G'_{(m+1|n)}\begin{pmat}({.|})
\Omega(\mathcal E_0,\mathcal O_1)&D(\mathcal O_1^\dagger)&\Omega(\mathcal E_1,\mathcal O_1)\cr
\mathcal E_0&0&\mathcal E_1\cr\-
\Omega(\mathcal E_0,\mathcal O_0)&D(\mathcal O_0^\dagger)&\Omega(\mathcal E_1,\mathcal O_0)\cr
\end{pmat},
\]
in an even $(m+1|n)\times(m+1|n)$-supermatrix.
\end{thm}
\begin{proof}
The formula for $U[m+n]$ in \eqref{A U} follows directly from \eqref{AU} and \eqref{Ber}.	
	
For any $i=0,\dots,m+n$,
\begin{equation}
	A[i]=
	\begin{vmatrix}
	\Omega(\mathcal E^i,\mathcal O^i)&\hat{\mathcal O}^i{}^\dagger\\
	\mathcal E^i&\fbox{$0$}
    \end{vmatrix},
\end{equation}
and
\begin{equation}
	\theta[i]=
	\begin{vmatrix}
	\Omega(\mathcal E^i,\mathcal O^i)&\Omega(\theta_i,\mathcal O^i)\\
	\mathcal E^i&\fbox{$\theta_i$}
	\end{vmatrix},\quad
	\rho[i]=\begin{vmatrix}
	\hat{\Omega(\mathcal E^i,\mathcal O^i)}&\mathcal O^i{}^\dagger\\
	\hat{\Omega(\mathcal E^i,\rho_i)}&\fbox{$\rho_i$}
\end{vmatrix},
\end{equation}
where $A[0]=U[0]=0$ and $\theta[0]=\theta_0$, $\rho[0]=\rho_0$.
Also define
\begin{equation}
\Omega[i]=\Omega(\theta[i],\rho[i])=
\begin{vmatrix}
    \Omega(\mathcal E^{i},\mathcal O^{i})&\Omega(\theta_i,\mathcal O^i)\\
    \Omega(\mathcal E^i,\rho_i)&\fbox{$\Omega(\theta_i,\rho_i)$}
\end{vmatrix}.
\end{equation}
Using \eqref{Ber},
\begin{equation}
	\Omega[i]=
	\begin{cases}
		\dfrac{\Ber(\Omega(\mathcal E^{i+1},\mathcal O^{i+1}))}{\Ber(\Omega(\mathcal E^{i},\mathcal O^{i}))}&0\le i\le m-1 \\[6pt]
		\dfrac{\Berstar(\Omega(\mathcal E^{i+1},\mathcal O^{i+1}))}{\Berstar(\Omega(\mathcal E^{i},\mathcal O^{i}))}&m\le i\le m+n-1.
	\end{cases},
\end{equation}
where $\Ber(\Omega(\mathcal E^{0},\mathcal O^{0}))=1$. If fact, for $1\le i\le m$ then $\Ber(\Omega(\mathcal E^{i},\mathcal O^{i}))=\det(\Omega(\mathcal E^{i},\mathcal O^{i}))$,

Since all pairs $(\theta_j,\rho_j)$ have opposite parity, $\theta[i]$ and $\rho[i]$ also have opposite parity and so commute with each other. Hence \eqref{A[k]} may be written for arbitrary $i=0,\dots,m+n-1$ as
\begin{equation}
	A[i+1]=A[i]-\frac{\hat{\rho[i]}\theta[i]}{\Omega[i]}.
\end{equation} 	
Now, the parity of $\rho[i]$ is the same as that of $\rho_i$ and hence
\begin{equation}
\hat{\rho[i]}=
\begin{cases}
	-\rho[i]&0\le i\le m-1\\	
	\rho[i]&m\le i\le m+n-1
\end{cases}.
\end{equation}
Finally, this gives the recurrence relation
\begin{equation}
	A[i+1]=A[i]+\begin{cases}
		D(\log\Omega[i])&0\le i\le m-1\\	
		-D(\log\Omega[i])&m\le i\le m+n-1
	\end{cases},\quad A[0]=0,
\end{equation} 	
and so
\begin{equation}
	A[m+n]=D\left(\log\frac{\Omega[0]\Omega[1]\cdots\Omega[m-1]}{\Omega[m]\Omega[m+1]\cdots\Omega[m+n-1]}\right)=D\bigl(\log\Ber(\Omega(\mathcal E,\mathcal O))\bigr),
\end{equation}
as required.
\end{proof}

\begin{rem}
The earlier papers on this topic \cite{ShTu,LM} take $n=0$ only. In this case, $\mathcal E=\mathcal E_0$ and $\mathcal O=\mathcal O_1$ and we obtain solutions expressed in terms of determinants not superdeterminants
\[
A[m]=D(\log\det\Omega(\mathcal E_0,\mathcal O_1)),
\]
and
\[
U[m]=\frac{\det\begin{pmatrix}
    \Omega(\mathcal E_0,\mathcal O_1)&D(\mathcal O_1^t)\\
    \mathcal E_0&0
\end{pmatrix}}{\det(\Omega(\mathcal E_0,\mathcal O_1))}.
\]
\end{rem}

\section{Conclusions} In this paper, we have considered a twisted derivation which includes normal derivative, forward difference operator, q-difference operator and superderivatives as special cases. We showed that a Darboux transformation defined in terms of such a twisted derivation has an iteration formula written in terms of a quasideterminant. This results opens the opportunity for an unified approach to Darboux transformations for differential, superdifferential difference and $q$-difference operators. In this paper we showed how this was achieved for one example, the Manin-Radul super KdV equation.

In the Darboux transformation approach to this equation given in \cite{QPM}, the authors were forced to consider two different cases, for an odd and an even number of iterations, in order to obtain solutions expressed in terms of superdeterminants. Using the same Darboux transformation, we have obtained quasideterminant solutions in a unified manner, irrespective of the parity of the number of iterations, and then we were able in all cases to reexpress the solutions in terms of superdeterminants using the known relationship between quasideterminants and superdeterminants. This illustrates the advantage of using quasideterminants over using superdeterminants from the start and comes about because quasideterminants need no assumption about the nature of the noncommutativity whereas for superdeterminants restrictive assumptions about parity of matrix elements are required. The same advantage holds in connection with iterated binary Darboux transformations. In this case we were able to construct a much more general types of solution as well. In earlier work \cite{ShTu,LM} only solutions expressed in terms of determinants were found. However, it is to be expected that for a supersymmetric integrable system the most general solutions have expressions in terms of superdeterminants rather than determinants. This deficiency was remedied in this paper and we have obtained a much wider class of solutions expressed in terms of both quasideterminants and superdeterminants.

\section*{Acknowledgement} This work was carried out during the authors' visit to the Newton Institute for Mathematical Sciences as part of the Programme on Discrete Integrable Systems (January-July 2009). The authors are grateful to the organizers for their invitation and to the Isaac Newton Institute for financial support and hospitality. This work was supported by the National Natural Science Foundation of China (grant No. 10601028) and the Project-sponsored by SRF for ROCS, SEM.

\end{document}